\documentclass[conference]{IEEEtran}
\IEEEoverridecommandlockouts
\usepackage{cite}
\usepackage[colorlinks=true, linkcolor=red,citecolor=orange,urlcolor=black]{hyperref}
\usepackage{subcaption}
\usepackage{caption}
\usepackage{amsmath,amssymb,amsfonts,amsthm}
\usepackage{array}
\usepackage{algpseudocode} 
\usepackage{algorithm}
\usepackage{algorithmicx}
\usepackage{graphicx}
\usepackage{url}
\usepackage{bm}   
\usepackage{amsmath}
\usepackage{epstopdf}
\usepackage{verbatim}
\usepackage{xcolor}
\usepackage{enumitem}
\usepackage{wrapfig}
\usepackage{balance}
\usepackage{braket}
\usepackage{physics}
\usepackage{siunitx}
\usepackage{lipsum}  
\usepackage{tikz}
\usepackage{pgfplots}
\usepackage{mathtools}
\usepackage{hyperref}
\usepackage{titlesec}
\usepackage{siunitx}

\graphicspath{{figures/}}
\def\BibTeX{{\rm B\kern-.05em{\sc i\kern-.025em b}\kern-.08em
    T\kern-.1667em\lower.7ex\hbox{E}\kern-.125emX}}

\usepackage{listings}
\usepackage{color}

\definecolor{dkgreen}{rgb}{0,0.6,0}
\definecolor{gray}{rgb}{0.5,0.5,0.5}
\definecolor{mauve}{rgb}{0.58,0,0.82}
\begin{document}

\newtheorem{thm}{Theorem}
\newtheorem{lma}{Lemma}
\newtheorem{defi}{Definition}
\newtheorem{proper}{Property}

\title{Distributed Quantum Magnetic Sensing for Infrastructure-free Geo-localization}
\author{
\IEEEauthorblockN {
   Thinh Le, 
   Shiqian Guo, 
   Jianqing Liu
}

\IEEEauthorblockA {
  Department of Computer Science, North Carolina State University, Raleigh, USA 27606.\\
}
\IEEEauthorblockA {
    \{tvle2, sguo26, jliu96\}@ncsu.edu
}
}

\IEEEcompsoctitleabstractindextext{%
\begin{abstract}

Modern navigation systems rely heavily on Global Navigation Satellite Systems (GNSS), whose weak spaceborne signals are vulnerable to jamming, spoofing, and line-of-sight blockage. As an alternative, the Earth's magnetic field entails location information and is found critical to many animals' cognitive and navigation behavior. However, the practical use of the Earth's magnetic field for geo-localization hinges on an ultra-sensitive magnetometer. This work investigates how quantum magnetic sensing can be used for this purpose. We theoretically derive the Cramér--Rao lower bound (CRLB) for the estimation error of quantum sensing when using a nitrogen-vacancy (NV) center and prove the quantum advantage over classical magnetometers. Moreover, we employ a practical distributed quantum sensing protocol to saturate CRLB. Based on the estimated magnetic field and the earth's magnetic field map, we formulate geo-localization as a map-matching problem and introduce a coarse-to-fine Mahalanobis distance search in both gradient space (local field derivatives) and corner space (raw field samples). We simulate the proposed quantum sensing-based geo-localization framework over four cities in the United States and Canada. The results report that in high-gradient regions, gradient-space Mahalanobis search achieves sub-kilometer median localization error; while in magnetically smoother areas, corner-space search provides better accuracy and a $4-8\times$ reduction in runtime. 
\end{abstract}
}

\maketitle
 
\thispagestyle{plain}
\IEEEdisplaynotcompsoctitleabstractindextext
\IEEEpeerreviewmaketitle

\section{Introduction}
Robust and reliable navigation is a cornerstone of modern technology, supporting applications ranging from autonomous transportation to defense systems \cite{nahavandi2025comprehensive, olthuis2025cyberattacks}. For decades, the Global Positioning System (GPS) and similar satellite constellations have served as the primary sources of global positioning. However, the reliance on these weak signals beamed from space introduces critical vulnerabilities. They are susceptible to jamming, spoofing, and are simply unavailable in environments such as underwater, underground, or in large buildings \cite{zidan2020gnss, khan2021gps, radovs2024recent, tanaka2023first}. This fragility necessitates the development of complementary navigation technologies that remain robust in contested environments.

One promising approach for navigation and localization is to use the Earth's magnetic field \cite{lei2025review, karshakov2020promising}. The Earth's core, which is composed primarily of molten iron, acts as a colossal dynamo that produces our planet's magnetic field. This interplay of iron and motion deep within the Earth generates magnetic fields strong enough to influence the surface, commonly referred to as the Earth's main (core) magnetic field. Superimposed on this main field are magnetic anomalies. These irregularities are often caused by different geological formations and local variations in crustal magnetization. One significant contribution comes from magnetic minerals embedded within the Earth's crust, which create localized disruptions that reveal hidden geological structures. 
While the earth's magnetic field entails valuable geo-spatial structure, which many marine animals (e.g, sea turtles, whales) and migratory birds use for migration \cite{packmor2024migratory, keller2021map, wiltschko2022discovery}, the ultra-weak and subtle variations of the Earth's core and anomaly magnetic field present unique challenges for civilian and military usage. With the development of highly sensitive quantum sensors, using the Earth’s magnetic field, specifically the magnetic anomalies, as a passive, globally available reference for navigation and localization may become a viable technical solution.

Realizing this potential requires magnetometers with sufficient sensitivity and stability to resolve weak crustal signatures on top of a much larger background field \cite{wang2022quantum, lonard2025limits, barry2020sensitivity}. The nitrogen–vacancy (NV) center in a diamond is an excellent candidate, formed by replacing one carbon atom in the lattice with a nitrogen atom and creating an adjacent vacant lattice site \cite{li2025review}. The negatively charged NV-center is sensitive to the ambient magnetic field. It features an electronic spin-triplet ground state with three sub-levels labeled by the spin projection $m_S = \{0, \pm 1\}$, which can be initialized and read out optically. In a typical setup for magnetic sensing, the spin is first prepared in the $m_S = 0$ state and then driven to coherent transitions (or superposition) between $m_S = 0$ and $m_S = \pm 1$. An external magnetic field shifts the corresponding resonance frequencies (or the probability amplitudes between $m_S$'s) via the Zeeman effect. By monitoring the frequency, one can estimate the strength of the magnetic field with high precision \cite{kuwahata2020magnetometer, hong2013nanoscale, abe2018tutorial}. This setup is a standard Ramsey interferometry for a single shot measurement \cite{degen2017quantum}. Like classical sensing paradigms, $N$ independent executions of NV-center Ramsey measurements only lead to $\frac{1}{\sqrt{N}}$ error scaling due to the independent per-shot Gaussian noises. This limit is widely known as the standard quantum limit (SQL). In contrast, some recent works have demonstrated that NV-center quantum magnetometers with only correlated measurements (without using entanglement) can beat SQL and approach the Heisenberg limit (HL) with a scaling factor of $\frac{1}{N}$ \cite{guo2025two,bonato2016optimized}, which is the ultimate precision limit only constrained by the law of physics. Therefore, it enables highly sensitive sensing of the Earth’s magnetic field and results in accurate geo-localization.


In a brief overview, this paper makes the following contributions:
\begin{itemize}
    \item We derive the Cramér–Rao lower bound (CRLB) for a single-spin NV-center magnetometer operated via Ramsey interferometry and demonstrate the quantum advantage over classical scalar magnetometers.
    \item Building on our prior work, we adopt a two-stage NV-center sensing protocol proven to beat SQL and show how it naturally provides field estimates necessary for geo-localization.
    \item We formulate geo-localization as a distributed sensing and map-matching problem. We introduce a Mahalanobis distance search algorithm in both gradient space and corner space to localize the site of interest.
    \item Using the U.S. Geological Survey (USGS) magnetic anomaly grid for North America, we simulate the proposed design over four geographically diverse sites in the contiguous United States (CONUS) and demonstrate high precision in localization.
\end{itemize}

The remainder of this paper is organized as follows. $\S$~\ref{sec:related_works} reviews prior work on magnetic anomaly navigation. Next, $\S$~\ref{sec:problem} formulates the geo-localization problem. Following that, $\S$~\ref{sec:discussion} develops the Cramér–Rao bound analysis for both classical magnetometers and the single-spin NV-center magnetometer. $\S$~\ref{sec:methods} describes the proposed technique to tackle the previously presented research problem. $\S$~\ref{sec:simulation} presents our simulation results. Finally, $\S$~\ref{sec:conclusion} summarizes our findings and outlines directions for future work.

\section{Related Work}
\label{sec:related_works}
The realization of magnetic-based geo-localization systems rests on three key pillars: (1) high-resolution magnetic anomaly maps, (2) advanced sensors capable of measuring the magnetic field with nanotesla-level sensitivity, and (3) a software stack for extracting signals and performing localization. In this section, we review recent advances across these pillars that enable magnetic-based localization and navigation.

The foundation of any magnetic-based localization system is the magnetic anomaly map, which provides the reference field for positioning. Gupta et al.~\cite{gupta2024lower} established the relationship between magnetic anomaly map quality and achievable positioning accuracy via the CRLB. To implement the bound, they considered two measurement noise models: unimodal Gaussian noise and multimodal noise. Through simulations over a $250 \times 400~\text{km}^2$ region, their results demonstrated a strong correspondence between higher magnetic gradients and lower CRLB (i.e., better accuracy). Recognizing that localization performance depends on the spatial variability of the magnetic field, Penumarti et al.~\cite{penumarti2024global} addressed localization uncertainty in magnetic anomaly navigation by introducing a global, uncertainty-aware path planner. Their key contribution was the generation of an entropy map from pre-existing magnetic anomaly map to statistically quantify areas of high information gain. They employed a multi-objective potential field planner that attracted the vehicle toward these low-entropy regions while simultaneously guiding it to a global goal. Hardware experiments using TurtleBot 4 equipped with a Quspin QTFM total-field magnetometer demonstrated that this method produced paths that significantly improved localization stability and reduced estimation covariance compared to existing methods.

While map quality is important, accurate measurement of weak magnetic anomalies requires highly sensitive magnetometers. Bonato et al.~\cite{bonato2016optimized} demonstrated DC magnetic field sensing with single-spin NV-center magnetometer using a real-time adaptive measurement protocol. They implemented an adaptive Ramsey interferometry protocol in which the spin readout basis was dynamically updated after each single-shot measurement using a Bayesian estimator. The protocol achieved a remarkable magnetic field sensitivity of $6.1 \pm 1.7 ~\text{nT}/\sqrt{\text{Hz}}$ over a dynamic range of $1.78~\text{mT}$. Building on the development of high-sensitivity NV-center magnetometers, Wang et al.~\cite{wang2022quantum} investigated their viability as aiding sensors for INS using total magnetic intensity maps for localization. However, in practical platforms, the weak crustal magnetic signal is often corrupted by vehicle-induced interference, making signal separation a critical challenge. To improve magnetic signal fidelity, the United States Air Force Research Laboratory, in collaboration with MIT Lincoln Lab, introduced the MagNav signal enhancement challenge \cite{gnadt2020signal}, which aimed to develop algorithms to separate the Earth’s magnetic anomaly field from strong aircraft-induced magnetic interference. A team from SandboxAQ addressed the challenge by proposing a physics-informed machine learning model that integrated the classical Tolles-Lawson (T-L) model with a Liquid Time-Constant Network (LTC) \cite{nerrise2024physics}. Their approach used a closed-form continuous-time (CfC) variant of the LTC to model the nonlinear dynamics of the aircraft-induced magnetic noise. Experimental results demonstrated that the proposed LTC-CfC network achieved up to a $64\%$ reduction in compensation error compared to the standard T-L model and outperformed other baseline methods.

Building on these foundations, several integrated systems have demonstrated practical MagNav capabilities. Q-CTRL introduced a full-stack system~\cite{muradoglu2025quantum} integrating a quantum scalar optically pumped atomic magnetometer, a classical fluxgate vector magnetometer, and a software stack for magnetic denoising and map matching. The software stack processes the magnetometer readings to perform real-time platform denoising and provides periodic position fixes by matching the cleaned magnetic signal against a known anomaly map (e.g., EMAG2v3). Across flight trials, this quantum-assured system consistently outperformed a strategic-grade INS, achieving superior bounded positioning accuracy with a best-case result of 22 meters. In ground trials, the system demonstrated a bounded positioning error approximately seven times lower than that of the INS. However, a major limitation of such map-based approaches is their dependence on pre-surveyed maps. To address this, Lee et al.~\cite{lee2020magslam} introduced MagSLAM, an aerial simultaneous localization and mapping framework (SLAM) that exploits the Earth's magnetic anomaly field without relying on pre-surveyed magnetic maps. The authors adapted the FastSLAM algorithm (Rao-Blackwellized particle filter)~\cite{montemerlo2002fastslam} to fuse scalar magnetic anomaly measurements with an INS. Their method uses a nine-state INS error model to constrain drift, while simultaneously estimating the vehicle's trajectory and building a local magnetic map without any prior survey. Flight tests were performed with a Cessna 208B equipped with a cesium-vapor scalar magnetometer, barometer, and high-grade INS over a $9 \times 12~\text{km}$ region in Louisa, Virginia, at $150~\text{m}$ above ground level. The MagSLAM filter achieved an approximate $17~\text{m}$ distance root mean square position accuracy over a 100-minute flight.

\section{Problem formulation}
\label{sec:problem}

In this section, we study the problem of geo-localization using NV-center quantum measurements of the earth's magnetic field. We consider a platform that moves within a prescribed region of interest and is equipped with one or more magnetometers. A reference magnetic map of the region (typically a magnetic anomaly map) is assumed to be known. At each location, the platform executes a sensing protocol that produces noisy estimates of the local magnetic field. The localization task is to infer the unknown position from these noisy field estimates and the known magnetic map.

Let $d \ge 1$ and let $\mathcal{D} \subset \mathbb{R}^d$ denote the region of interest in which the platform may be located. We model the scalar magnetic field as follows:
\begin{equation}
    B : \mathcal{D} \to \mathbb{R},
    \qquad x \mapsto B(x),
\end{equation}
where $B(x)$ denotes the magnitude of the local geomagnetic field at position $x$. If the true position $x$ were known, the localization engine could query the map and evaluate $B(x)$, obtaining the ideal magnetic field value at that location. In localization, we face the inverse situation: we observe noisy measurements of $B(x)$ at an unknown position and attempt to infer $x$ from these samples, together with the known map. The true platform position is modeled as an unknown variable:
\begin{equation}
    x^\star \in \mathcal{D}.
\end{equation}

\noindent\textbf{Sensor model.} We introduce a general sensor model that applies to the NV-center magnetometer considered in this work. Consider a platform comprising $S \ge 1$ sensors, each producing a distinct position-dependent output. For each sensor $s \in \{1,\dots,S\}$, we model the ideal output at position $x \in \mathcal{D}$ by
\begin{equation}
    h_s : \mathcal{D} \to \mathbb{R},
    \qquad x \mapsto h_s(x).
\end{equation}
In the simplest case of a single sensor co-located with the platform,
\begin{equation}
    S = 1,
    \qquad
    h_1(x) = B(x),
\end{equation}
so the ideal sensor output equals the local magnetic field. For each sensor $s$, the function $h_s$ specifies the ideal measurement produced by that sensor at any position $x$ in the region. Different sensing strategies are then distinguished only by how they choose measurement settings (e.g., number of measurement shots, per-shot sensing times) and how they process the resulting noisy observations.

\noindent\textbf{Measurement model.} At the unknown position $x^{\star}$, the sensing protocol produces a the scalar magnetic field estimate $\hat{B}_s$ for each sensor $s$. We model these estimates as: 
\begin{equation}
    \hat{B}_s = h_s(x^\star) + \varepsilon_s,
    \label{eq:meas-basic-hatB}
\end{equation}
where $\varepsilon_s$ is an effective noise term summarizing both hardware noise and the internal estimation error of the sensing protocol. We assume
\begin{equation}
\mathbb{E}[\varepsilon_s] = 0,
\end{equation}
and the noise terms $\{\varepsilon_{s}\}_{s=1}^S$ are independent across sensors. Collecting all field estimates into a single vector
\begin{equation}
    \hat{B} \coloneqq
    \begin{bmatrix}
        \hat{B}_1 & \cdots & \hat{B}_S
    \end{bmatrix}^\top
    \in \mathbb{R}^S,
\end{equation}
we can write the model compactly as
\begin{equation}
    \hat{B} = h(x^\star) + \varepsilon,
    \label{eq:vector-meas-hatB}
\end{equation}
where $h : \mathcal{D} \to \mathbb{R}^S$  is the stacked ideal sensor outputs, and $\varepsilon \in \mathbb{R}^S$ is the corresponding noise vector. Eq.~\eqref{eq:vector-meas-hatB} is the fundamental forward model in which the unknown state $x^\star$ is mapped to an ideal measurement vector $h(x^\star)$ and then corrupted by additive noise $\varepsilon$. The localization problem is to invert this mapping as accurately as possible given a single realization of $\hat{B}$.

\noindent \textbf{Localization problem.} Given the known magnetic map, the unknown position $x^\star \in \mathcal{D}$, and a measurement model of the form Eq.~\eqref{eq:meas-basic-hatB} or Eq.~\eqref{eq:vector-meas-hatB}, the magnetic localization task is to reconstruct $x^\star$ from the observed data. A localization rule is a mapping
\begin{equation}
    \hat{x} : \mathbb{R}^S \to \mathcal{D},
\end{equation}
which takes as input the vector of field estimates $\hat{B}$ produced by the sensing system and outputs an estimated position $\hat{x}(\hat{B})$.

The measurement noise term $\varepsilon$ in Eq.~\eqref{eq:vector-meas-hatB} represents the critical performance bottleneck where quantum sensors can provide substantial advantages. Classical magnetometers are limited by thermal noise, flicker noise, and sensor drift, while quantum magnetometers, such as NV-center sensors, operate under different fundamental noise mechanisms. In the next section, we compare the fundamental sensitivity limits of NV-center and classical magnetometers using Cramér–Rao lower bound framework.


\begin{figure}[t]
\centering
    \includegraphics[width=\columnwidth]{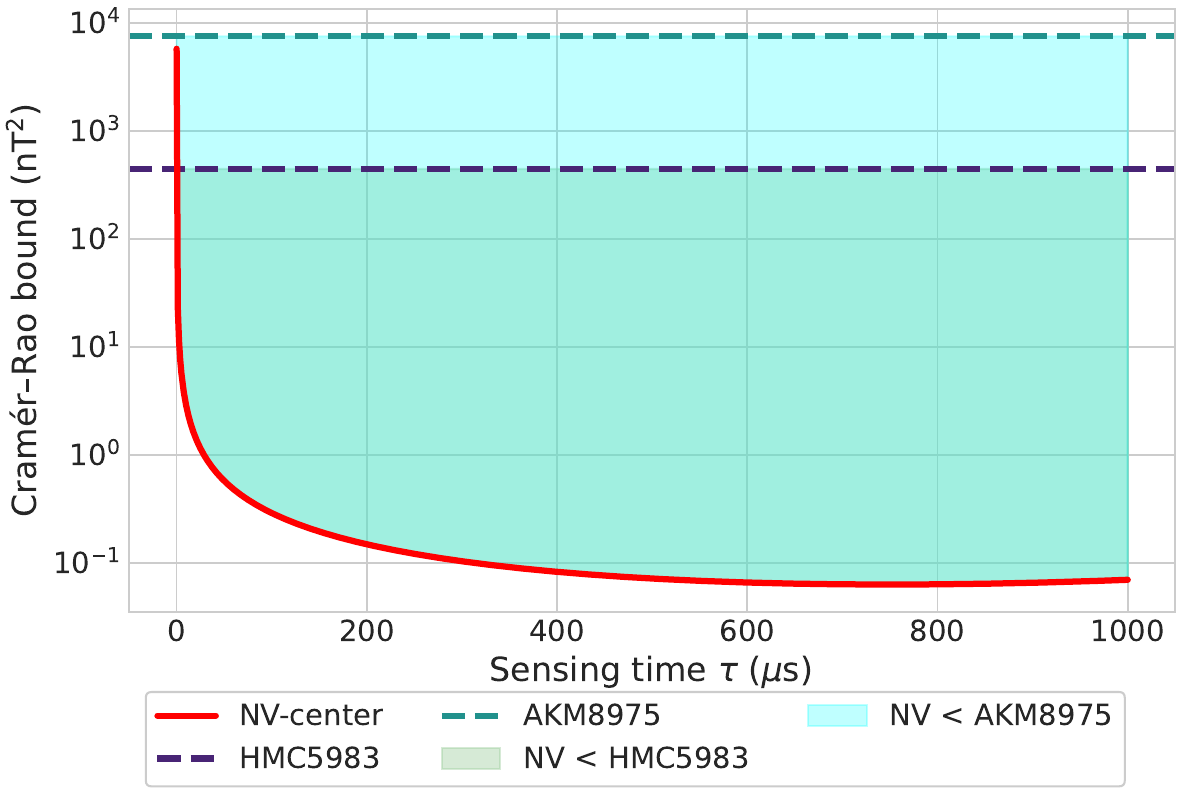}
    \caption{The CRLBs for a single-spin NV-center magnetometer and two commercial sensors (HMC5983 and AKM8975)~\cite{dammann2024cramer}, with respect to sensing time per measurement shot $\tau$. The red curve shows the NV-center CRLB after a fixed total sensing time $T_{\mathrm{total}} = 1.5~\mathrm{s}$ with $T_2^\ast = 1500~\mu\mathrm{s}$ and readout fidelities $F_0 = 0.88$ and $F_1 = 0.98$.}
    \label{fig:qt-advan}
\end{figure}
\section{The Cramér-Rao Lower Bound Analysis \label{sec:discussion}}

CRLB in the area of estimation theory defines the absolute minimum variance that any unbiased estimator can have when estimating an unknown parameter of interest. In this section, we first describe the single-spin NV-center quantum sensing setup, then derive the corresponding CRLB on magnetic field estimation, and finally compare it to the CRLB of classical scalar magnetometers. In so doing, we will demonstrate quantum advantage of quantum sensing over classical apparatus in the task of earth magnetic field sensing. 

\subsection{NV-Center Quantum Sensing Model}
In this work, we consider a single NV-center electron spin in a diamond. 
Coherent manipulation of the electron spin states ($\ket{0} = \ket{m_s = 0}$ and $\ket{1} = \ket{m_s = -1}$) is implemented using resonant microwave pulses generated by the microwave system. The sensing protocol is based on Ramsey interferometry. The process begins with the initialization of the electron spin into the $\ket{0}$ state via resonant optical excitation. A subsequent $\pi/2$ microwave pulse prepares the spin in the superposition state $\ket{\psi} = \frac{1}{\sqrt{2}} (\ket{0} + \ket{1})$ with density matrix: 
\begin{equation}
    \rho = |\psi\rangle \langle \psi| = \frac{1}{2} \begin{pmatrix} 1 & 1 \\ 1 & 1 \end{pmatrix}
\end{equation}
This state then evolves freely for a sensing time $\tau$. During this evolution, two processes occur simultaneously:
\begin{enumerate}
    \item \textbf{Unitary evolution:} the Zeeman shift from the external magnetic field causes the state to accumulate a relative phase $\varphi = \omega \tau$, where $\omega$ is the Larmor angular frequency of the NV spin.
    \item \textbf{Decoherence (dephasing):} The spin state is subject to $T_2^*$ decoherence due to environment noises such as the surface charges in the diamond. The decoherence can be modeled as a Gaussian decay of the off-diagonal elements of the density matrix.
\end{enumerate}
Combining these effects, the density matrix of the spin state before measurement becomes:
\begin{equation}
    \rho(\omega, \tau) = \frac{1}{2} \begin{pmatrix}
1 & e^{-(\tau/T_2^*)^2} e^{i\omega \tau} \\
e^{-(\tau/T_2^*)^2} e^{-i\omega \tau} & 1
\end{pmatrix}
\end{equation}
Following the evolution period, a second $\pi/2$ pulse with a controllable phase is applied to convert the accumulated phase into a population difference. Finally, spin-state readout is realized through resonant spin-selective optical excitation, with emitted photons detected in two orthonormal basis $\ket{0}$ and $\ket{1}$. The readout fidelities denoted as $F_0$ (e.g., $0.88 \pm 0.02$) for the $\ket{0}$ basis and $F_1$ (e.g., $0.98 \pm 0.02$) for the $\ket{1}$ basis \cite{bonato2016optimized} capture the readout noises.

\subsection{Quantum Cramér–Rao Lower Bound}\label{subsec:qtcrb}
To determine the fundamental precision limit of the sensing protocol, we must quantify the maximum information about the Larmor angular frequency $\omega$ that can be extracted from experimental measurement outcomes. The experimentally validated probability of measuring the spin in state $\ket{0}$ after a single Ramsey interferometry is modeled by \cite{bonato2016optimized}:
\begin{align}\label{eq:likelihood}
    p_0 = P(0|\omega) &= \frac{(1 + F_0 - F_1)}{2} \\
    &\quad + \frac{(F_0 + F_1 - 1)}{2} 
    e^{-\left(\frac{\tau}{T_2^*}\right)^2} 
    \cos\left[\omega \tau + \vartheta\right]. \nonumber
\end{align}
where $F_0$ and $F_1$ are the readout fidelities of states $\ket{0}$ and $\ket{1}$, respectively. Here, $\tau$ is the sensing time, $T_2^*$ is the spin dephasing time, and $\vartheta$ is the controllable phase of the second $\pi/2$ microwave pulse in the Ramsey sequence. For convenience, we introduce two constants:
\begin{align}
    \alpha = \frac{(1 + F_0 - F_1)}{2}, \quad \beta = \frac{(F_0 + F_1 - 1)}{2}. 
\end{align}

To quantify the gained knowledge from this measurement, we use the Fisher information. The Fisher information measures how the outcome probabilities respond to changes in the parameter of interest, $\omega$. In our case, the Fisher information, and thus the measurement precision, depends on the controllable phase $\vartheta$ of the second $\pi/2$ pulse and is given as follows \cite{nolan2021machine}:
\begin{align}\label{eqn:cF}
    F(\omega) &= \frac{1}{p_0} \left( \frac{\partial p_0}{\partial \omega} \right)^2 + \frac{1}{1 - p_0}\left( \frac{\partial (1-p_0)}{\partial \omega} \right)^2\\
            &= \frac{1}{p_0 (1-p_0)} \left( \frac{\partial p_0}{\partial \omega} \right)^2. \nonumber
\end{align}
To determine the sensitivity of the measurement probability to changes in frequency, we differentiate $p_0$ with respect to $\omega$ as follows:
\begin{align}
    \frac{\partial p_0}{\partial \omega} &= \beta  e^{-\left(\frac{\tau}{T_2^*}\right)^2} \frac{\partial}{\partial \omega} (\cos[\omega \tau + \vartheta]) \\ \nonumber
                                      &= -\tau  \beta  e^{-\left(\frac{\tau}{T_2^*}\right)^2} \sin[\omega \tau + \vartheta].
\end{align}
The squared derivative is therefore:
\begin{equation}\label{eqn:derive}
    \left(\frac{\partial p_0}{\partial \omega}\right)^2 = \tau^2  \beta^2  e^{-2\left(\frac{\tau}{T_2^*}\right)^2} \sin^2[\omega\tau + \vartheta].
\end{equation}
By substituting Eq.~\eqref{eqn:derive} into Eq.~\eqref{eqn:cF}, we obtain the Fisher information as:
\begin{equation}\label{eqn:cf_new}
    F(\omega) = \frac{\tau^2  \beta^2  e^{-2\left(\frac{\tau}{T_2^*}\right)^2} \sin^2[\omega\tau + \vartheta]}{p_0  (1-p_0))}.
\end{equation}
To establish a fundamental bound, we choose $\vartheta$ to maximize $F(\omega)$. This is equivalent to making the measurement at the point where the signal is most sensitive to a change in $\omega$, which corresponds to maximizing $\left(\frac{\partial p_0}{\partial \omega}\right)^2$. Maximum sensitivity occurs when $\sin^2[\omega \tau + \vartheta] = 1$, which implies that $\omega \tau + \vartheta = \frac{\pi}{2} + k \pi$ for any integer $k$. At this optimal point, we have $\cos[\omega \tau + \vartheta] = 0$, yielding the measurement probabilities:
\begin{align}
    p_0^{opt} &= \frac{(1 + F_0 - F_1)}{2},\quad p_1^{opt} =  \frac{(1 - F_0 + F_1)}{2}.
\end{align}
We find the maximum Fisher information that can be extracted by this measurement scheme is:
\begin{align}
    F^{max}(\omega) =  \tau^2\frac{(F_0 + F_1 - 1)^2}{1- (F_0 - F_1)^2} e^{-2\left(\frac{\tau}{T_2^*}\right)^2}.
\end{align}
Therefore, the quantum CRLB of any unbiased estimator of $\omega$ is given by:
\begin{align}
    Var(\omega) & \geq \frac{1}{F^{\text{max}}(\omega)} =
      \frac{1 - (F_0 - F_1)^2}{\tau^2 (F_0 + F_1 -1)^2} e^{2\left(\frac{\tau}{T_2^*}\right)^2}
\end{align}
A spin-based magnetometer can sense a DC magnetic field $B$ via Zeeman shift, where the Larmor angular frequency $\omega$ is related to the magnetic field by  $\omega = \gamma B$ with $\gamma$ being the gyromagnetic ratio. Propagating the bound through this relation yields the bound for the minimum estimation variance of an unknown magnetic field for using a single-spin NV-center magnetometer:
\begin{equation}
Var_Q(B) \geq \frac{1- (F_0 - F_1)^2}{\gamma^2 \tau^2 (F_0 + F_1 -1)^2} e^{2\left(\frac{\tau}{T_2^*}\right)^2}
\end{equation}

\subsection{Classical Magnetometer Cramér–Rao Lower Bound}\label{subsec:ccrlb}
To establish a baseline for the quantum sensing protocol, we first consider the fundamental precision limit of a classical scalar magnetometer. We model a single measurement at a fixed location as:
\begin{equation}
    \hat{B} = B + \epsilon
\end{equation}
where $B$ is the true magnetic field and $\epsilon$ is the additive Gaussian noise with zero mean and variance $\sigma^2$. The precision of any unbiased estimate of the magnetic field based on this measurement is fundamentally limited by the CRLB, which is derived from the Fisher information of the measurement's probability distribution. The result is formally stated in the following theorem.
\begin{thm}\label{thm:1}
    Given a single measurement $\hat{B}$ from a classical magnetometer where the noise is additive and Gaussian with variance $\sigma$, the Cramér–Rao bound $Var_C(B)$ of any unbiased estimator of the magnetic field $B$ is bounded by: 
    \begin{equation} Var_C(B) \ge \sigma^2 \end{equation}
\end{thm}
\begin{proof}
    The proof is detailed in Appendix \ref{appen:A}. 
\end{proof}

\begin{figure*}[t]
	\centering
	\includegraphics[width=\textwidth, height=2.2 in, keepaspectratio]{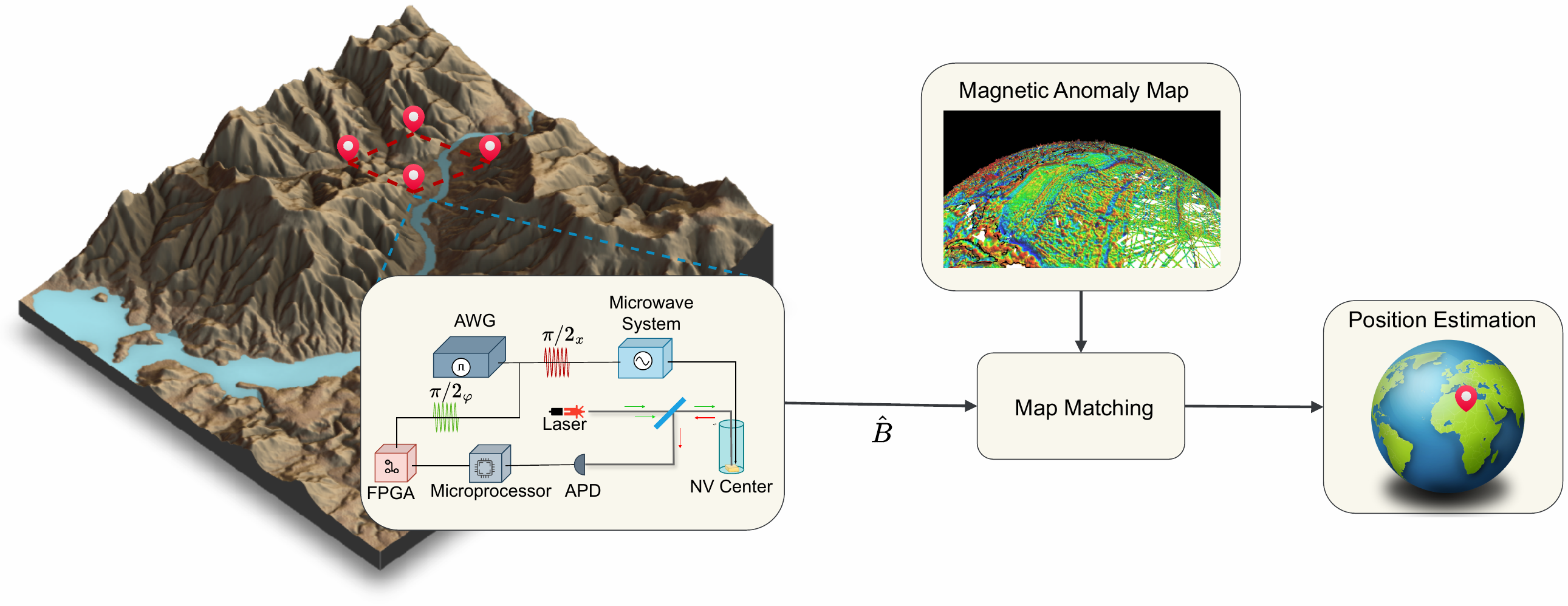}
	\caption{ Overview of quantum magnetic sensing and localization system. A single-spin NV-center magnetometer converts local magnetic fields into field estimates $\hat{B}$, which are fed to a map matching engine together with a magnetic anomaly map. The map matching module compares $\hat{B}$ to the map and outputs a position estimate.}
    \label{fig:summary}
\end{figure*}

To demonstrate the quantum advantage of the NV-center magnetometer, we compare its CRLB against that of established classical magnetometers. Fig.~\ref{fig:qt-advan} shows the quantum CRLB on magnetic field sensing as a function of sensing time $\tau$ per Ramsey shot, for a fixed total sensing time $T_{total} = 1.5~\mathrm{s}$. In comparison, we set the inhomogeneous dephasing time $T_2^* = 1,500 ~\mu\mathrm{s}$, consistent with room-temperature single-spin systems~\cite{herbschleb2019ultra}.
The classical sensors used for this benchmark were selected from \cite{dammann2024cramer}, which are neither dominated by external magnetic field noise (e.g., MMC246xMT) nor exhibit anisotropic resolution (e.g., YAS532, YAS529). For each device, we use the noise variance reported in \cite{dammann2024cramer} as the variance parameter in the classical model of Theorem~\ref{thm:1}, which yields the horizontal lines in Fig.~\ref{fig:qt-advan}. As shown in figure, the quantum CRLB is highly dependent on sensing time, achieving a minimum at an optimal value before increasing due to decoherence. The shaded regions highlight the operational windows where quantum sensor's theoretical precision surpasses that of classical benchmarks, demonstrating quantum advantage. This trade-off motivates the need to operate the sensor in optimal sensing time.

\section{Methods}
\label{sec:methods}

In this section, we translate the abstract sensing and estimation framework into a concrete end-to-end localization pipeline. Building on the CRLB analysis in \S~\ref{sec:discussion} and the problem formulation in \S~\ref{sec:problem}, we show how single-spin NV-center magnetometers, an adaptive two-stage sensing protocol, and a Mahalanobis map matching engine are integrated into a practical system. We first present a high-level overview of the sensing and localization architecture, then describe the Bayesian neural network and reinforcement learning components that govern NV-center sensing, and finally detail the gradient-space and corner-space Mahalanobis search procedures that infer position from noisy field estimates and a precomputed magnetic map.

\subsection{System Overview}
Our magnetic localization system integrates three core components: (i) single-spin NV-center magnetometers, (ii) an adaptive magnetic field estimation protocol based on Bayesian Neural Network (BNN) and reinforcement learning (RL), and (iii) a map matching localization engine. As summarized in Fig.~\ref{fig:summary}, the NV sensors convert the local magnetic field into fluorescence measurements. The protocol turns measurements into field estimates at one or more spatial sampling locations. The map matching module then compares these estimates against precomputed map features to infer the position.

In this work, we deliberately adopt the same single-spin NV-center magnetometer setup as in our prior two-stage sensing protocol \cite{guo2025two}. The rationale for selecting this protocol is that we have experimentally demonstrated that it can saturate the quantum CRLB as illustrated in Fig.\ref{fig:qt-advan}, thus beating SQL and approaching HL. More importantly, it achieves the best estimation accuracy among existing methods in the literature for estimating wide-range magnetic field. 

\noindent\textbf{Two-stage adaptive sensing protocol.} Instead of operating the NV-center magnetometer with fixed settings, we employ a two-stage adaptive sensing protocol that dynamically selects the sensing time and control phase for each Ramsey interferometry under a finite time budget. Conceptually, this adaptive protocol sits between the hardware and the map matching engine. Given a sensing budget, the protocol selects a sequence of Ramsey measurements, processes the resulting fluorescence records, and produces a magnetic-field estimate $\hat{B}$. The collection of these estimates across all sites is then passed to the Mahalanobis map matching module, which treats them as noisy samples of the underlying magnetic map. A detailed description of the training and implementation of this two-stage protocol is provided in $\S$~\ref{subsec:two_stage}.

\noindent\textbf{Spatial sampling and map matching.}
At runtime, the localization system collects spatially distributed samples of the magnetic field at measurement sites within the region of interest. These sites may correspond either to multiple NV-center magnetometers deployed in close proximity for position inference or to a single sensor interrogating nearby locations to build a local magnetic fingerprint. The resulting collection of measurements defines a magnetic intensity pattern that can be compared against a reference magnetic anomaly map derived from satellite or aeromagnetic surveys. The reference map is typically stored as a GeoTIFF raster, where each cell encodes the magnetic field intensity, and spatial gradients computed from this raster form a distinctive magnetic signature for each location. The localization engine computes the statistical similarity between the measured pattern and the map features, identifying the cell whose magnetic signature best matches the observed data. This map matching procedure effectively inverts the magnetic field map by taking the measured field intensities and inferring the geospatial coordinates that would produce those readings. The resulting position estimate is reported in the map's coordinate reference frame and can be converted back to latitude and longitude for interpretation.

The components in this system work synergistically to achieve efficient localization. The adaptive sensing protocol optimizes each Ramsey measurement, producing refined field estimates that capture the relevant spatial structure of the magnetic field. The map matching engine then identifies the position that best explains these estimates in relation to the reference map.

\subsection{Two-Stage Adaptive Sensing Protocol}
\label{subsec:two_stage}

We adopt the two-stage adaptive sensing protocol introduced in our earlier study~\cite{guo2025two} and use it as the front end of the magnetic localization pipeline. We refer readers to \cite{guo2025two} for more details of the protocol while we only present a sketch of how it works here. The protocol comprises two phases: an offline training phase, followed by an online inference phase.

\noindent \textbf{Offline training.} In this phase, we trained two neural networks: a BNN estimator for coarse inference of Larmor frequency and a federated RL agent for adaptive refinement. For training the BNN, we specified the Larmor frequency interval $(0,\omega_{\max})$ that covers the signal of interest and discretized it into grid points ${\omega_j}$'s. A training dataset was generated by simulating Ramsey interferometry with fixed control parameters $ (\tau = \tau_{\min} = \pi / \omega_{\max}, \varphi = \varphi_0 = 0)$ for the likelihood function in Eq.(\ref{eq:likelihood}). Each simulated run produces a sequence of measurement outcomes $\vec{\mu}$, with the corresponding training label defined as the posterior distribution $p(\omega_j \mid \vec{\mu})$ over the discretized frequencies $\omega_j \in (0, \omega_{\max})$. The network was optimized by minimizing the cross-entropy loss between its output logits and the target posteriors, thereby learning to perform direct Bayesian inversion from experimental data to parameter estimates.

For the federated RL agent, the Larmor frequency range $(0,\omega_{\max})$ is partitioned into $N$ equal-width intervals $[\tfrac{n\,\omega_{\max}}{N}, \tfrac{(n+1)\,\omega_{\max}}{N})$,
each defining a local environment for the $n$-th RL agent. During training, after the \(k\)-th Ramsey interferometry, the agent’s input is the current posterior distribution \(p(\omega \mid \vec{\mu}^{(k)})\) restricted to the local subrange. It then outputs a control action $(\tau_k,\varphi_k)$, representing the sensing time and control phase for the next Ramsey interferometry. Given the resulting measurement outcome $\mu_{k+1}$, the posterior distribution is updated using Bayes' rule, and the Ramsey interferometry is repeated until the sensing budget is exhausted. Each local RL agent is trained to maximize its cumulative reward and periodically computes a local parameter update. These local updates are then aggregated using the FedAvg algorithm to produce a single global RL policy that generalizes across all subranges of \((0,\omega_{\max})\).

\noindent \textbf{Online inference.} After offline training, the sensing protocol moves to the two-stage online inference phase. In the non-adaptive stage, we simulated Ramsey interferometry with fixed sensing parameters $\tau_{\min} = \pi/\omega_{\max}$ and $\varphi_0 = 0$ to generate a sequence of single-shot measurement outcomes $\vec{\mu}$. The trained BNN estimator takes $\vec{\mu}$ as input and outputs a coarse estimate of the Larmor frequency $\hat{\omega}$, which defines the refinement interval \( \left[ \hat{\omega} - \frac{\Delta}{2}, \hat{\omega} + \frac{\Delta}{2} \right] \). In the adaptive stage, the posterior distribution within this subrange is approximated using a particle filter. The trained global RL agent utilizes this posterior distribution along with the remaining time resources to determine optimized sensing parameters $(\tau, \varphi)$ for the subsequent Ramsey interferometry. This iterative process of RL-based parameter selection and Bayesian updating continues until the sensing time budget $R_{\max}$ is exhausted.

\subsection{Mahalanobis Search}
\label{subsec:mahal}

Our map matching engine models localization problem as a maximum likelihood matching problem on the magnetic map. Given a noisy collection of magnetic field estimations from NV-center magnetometer, we search over candidate map cells and select the location whose magnetic signature is most consistent with the measurements. This is implemented as a coarse–to–fine Mahalanobis distance search, with two closely related metrics: (i) feature-space Mahalanobis distance based on local magnetic gradients, and (ii) corner-space Mahalanobis distance based directly on the measured field values. Both variants share the same search structure and differ only in their feature representation.

\noindent\textbf{Gradient-based feature.} We model the magnetic map as a scalar field $B(x,y)$ sampled on a uniform two-dimensional grid. After converting the magnetic map to planar coordinates (so that distances in \(x\) and \(y\) share a common length unit), the map
is represented as a 2D array
\begin{equation}
    B_{i,j} \coloneqq B(x_i,y_j),
\end{equation}
where $(x_i,y_j)$ denotes the planar coordinates of the grid point with integer indices $i,j$, and the grid spacing in the $x$ and $y$ directions is denoted by $h_x,h_y>0$. We define a $2\times 2$ \emph{cell} by its four neighboring pixels
\begin{equation}
(B^{\mathrm{LL}}_{i,j},\,B^{\mathrm{LR}}_{i,j},\,B^{\mathrm{UL}}_{i,j},\,B^{\mathrm{UR}}_{i,j})
\coloneqq
(B_{i,\,j+1},\,B_{i+1,\,j+1},\,B_{i,\,j},\,B_{i+1,\,j}),
\end{equation}
corresponding respectively to the lower-left (LL), lower-right (LR), upper-left (UL), and upper-right (UR) corners of the cell. For each $2\times 2$ cell, we compute a three-dimensional feature vector
\begin{equation}
    z_{i,j} \coloneqq
    \begin{bmatrix}
        g_x(i,j) \\[2pt] g_y(i,j) \\[2pt] d_{xy}(i,j)
    \end{bmatrix}
    \in \mathbb{R}^3,
\end{equation}
where $g_x$ and $g_y$ approximate the first-order spatial derivatives of $B$ and $d_{xy}$ approximates a mixed second derivative. These quantities are given by:
\begin{align}
    g_x(i,j) &=
    \frac{B^{\mathrm{LR}}_{i,j} + B^{\mathrm{UR}}_{i,j}
          - B^{\mathrm{LL}}_{i,j} - B^{\mathrm{UL}}_{i,j}}{2 h_x},
    \label{eq:gx-def}\\
    g_y(i,j) &=
    \frac{B^{\mathrm{UL}}_{i,j} + B^{\mathrm{UR}}_{i,j}
          - B^{\mathrm{LL}}_{i,j} - B^{\mathrm{LR}}_{i,j}}{2 h_y},
    \label{eq:gy-def}\\
    d_{xy}(i,j) &=
    \frac{B^{\mathrm{LL}}_{i,j} - B^{\mathrm{LR}}_{i,j}
          - B^{\mathrm{UL}}_{i,j} + B^{\mathrm{UR}}_{i,j}}{h_x h_y}.
    \label{eq:dxy-def}
\end{align}
In matrix form, if we collect the four corner values into
\begin{equation}
    b_{i,j} \coloneqq
    \begin{bmatrix}
        B^{\mathrm{LL}}_{i,j}\\[2pt]
        B^{\mathrm{LR}}_{i,j}\\[2pt]
        B^{\mathrm{UL}}_{i,j}\\[2pt]
        B^{\mathrm{UR}}_{i,j}
    \end{bmatrix},
\end{equation}
then the feature map is linear:
\begin{equation}
    \begin{bmatrix}
        g_x\\[2pt] g_y\\[2pt] d_{xy}
    \end{bmatrix}
    =
    T\,b_{i,j}, \qquad
    T =
    \begin{bmatrix}
        -\tfrac{1}{2h_x} & \tfrac{1}{2h_x} & -\tfrac{1}{2h_x} & \tfrac{1}{2h_x} \\
        -\tfrac{1}{2h_y} & -\tfrac{1}{2h_y} & \tfrac{1}{2h_y} & \tfrac{1}{2h_y} \\[2pt]
        \tfrac{1}{h_x h_y} & -\tfrac{1}{h_x h_y}
        & -\tfrac{1}{h_x h_y} & \tfrac{1}{h_x h_y}
    \end{bmatrix}.
\end{equation}
The vector $z_{i,j}$ characterizes the local shape of the magnetic field: the components $(g_x,g_y)$ encode the local slope of $B$ in the east–west and north–south directions, while $d_{xy}$ captures a mixed curvature term. Regions where the magnetic field changes rapidly correspond to large $||z_{i,j}||$, whereas magnetically smooth regions have small gradients and mixed derivatives. This representation emphasizes structurally distinctive regions with high magnetic contrast while remaining robust to global field variations.

\noindent\textbf{Mahalanobis distance in gradient-space.}
For a given candidate location, the applied quantum sensing protocol produces noisy estimates of the magnetic field at the four cell corners:
\begin{equation}
    \hat{b} =
    \begin{bmatrix}
        \hat{B}^{\mathrm{LL}} \\ \hat{B}^{\mathrm{LR}} \\ \hat{B}^{\mathrm{UL}} \\ \hat{B}^{\mathrm{UR}}
    \end{bmatrix},
    \quad
    \mathrm{Cov}(\hat{b}) = \Sigma_b \in \mathbb{R}^{4\times 4}.
\end{equation}
Here, $\Sigma_b$ is the covariance matrix of the magnetic field estimates and characterizes the uncertainty introduced by the NV sensing protocol (e.g., spin decoherence and readout noises). Using the same linear stencil, we map $\hat{b}$ into a measured feature vector
\[
    \hat{z} \coloneqq T \hat{b},
\]
By linear error propagation, the feature-space covariance is
\[
    \Sigma_z \coloneqq T \Sigma_b T^\top \in \mathbb{R}^{3\times 3}.
\]
For any map cell $(i,j)$ in the search region, we then define the gradient-space Mahalanobis distance
\begin{equation}
    D_{\mathrm{grad}}^2(i,j)
    \coloneqq
    \bigl(\hat{z} - z_{i,j}\bigr)^\top
    \Sigma_z^{-1}
    \bigl(\hat{z} - z_{i,j}\bigr).
    \label{eq:mahal-grad}
\end{equation}
The gradient-based search selects the cell that minimizes Eq.~(\ref{eq:mahal-grad}) over all candidates in the region of interest.

\noindent\textbf{Mahalanobis distance in corner-space.}
In addition to the gradient features, we also consider a complementary metric that uses directly the raw magnetic field measurements. Let $\mathcal{I} \subseteq \{\mathrm{LL},\mathrm{LR},\mathrm{UL},\mathrm{UR}\}$ denote the index set of observed corners, and let $S \coloneqq |\mathcal{I}|$. We collect the corresponding measurements into
\[
    \hat{b}^{(\mathcal{I})} \in \mathbb{R}^S, \qquad
    \Sigma_b^{(\mathcal{I})} \in \mathbb{R}^{S\times S},
\]
where $\Sigma_b^{(\mathcal{I})}$ is the per-corner covariance restricted to the observed entries. For each map cell $(i,j)$, we form the corresponding map vector $b_{i,j}^{(\mathcal{I})} \in \mathbb{R}^S$ by reading the magnetic field from the map at the same subset of corners. The corner-space Mahalanobis distance is then
\begin{equation}
    D_{\mathrm{corner}}^2(i,j)
    \coloneqq
    \bigl(\hat{b}^{(\mathcal{I})} - b_{i,j}^{(\mathcal{I})}\bigr)^\top
    \bigl(\Sigma_b^{(\mathcal{I})}\bigr)^{-1}
    \bigl(\hat{b}^{(\mathcal{I})} - b_{i,j}^{(\mathcal{I})}\bigr),
    \label{eq:mahal-corner}
\end{equation}
The multi-sensor localization use the same search procedure as in the gradient case but substitute Eq.~(\ref{eq:mahal-corner}) for Eq.~(\ref{eq:mahal-grad}). 

The two Mahalanobis distances target complementary aspects of the magnetic map, making each metric particularly effective for locations with distinct magnetic field patterns. The gradient-based distance addresses localization as a problem of matching the local shape of the magnetic field. This approach effectively cancels the dominant Earth's magnetic core field, which has large intensity (25-65 $\mu$T) and varies slowly over thousands of kilometers, while preserving the distinctive gradient signatures of local magnetic anomalies. By computing spatial derivatives, this metric serves as a high-pass filter that suppresses the nearly uniform background field and amplifies local magnetic contrasts. In contrast, the corner-based distance preserves the raw magnetic field intensities. This approach captures both magnitude and the correlation of the magnetic field measurements. As a result, it can be more informative in magnetically smooth regions, where spatial derivatives are small but the total field still exhibits meaningful spatial variation. 

\noindent\textbf{Mahalanobis search.}
Both metrics are embedded in the same coarse-to-fine search strategy, which uses the structure of the map to reduce computation and performs an exact search over a refined set of candidate cells. Let $\mathcal{R}$ denote an initial region of interest (ROI) in cell coordinates, such as the full map or a fixed-radius window around a prior. For each cell $(i,j)\in\mathcal{R}$ we can compute a distance
\[
    D^2(i,j) \in \{D_{\mathrm{grad}}^2(i,j),\, D_{\mathrm{corner}}^2(i,j)\},
\]
depending on whether we work in feature space or corner space. The search algorithm proceeds as follows:

\begin{algorithm}[h]
\caption{Coarse-to-fine Mahalanobis map matching}
\label{alg:mahalanobis-search}
\begin{algorithmic}[1]
\Statex \textbf{Input:} feature map $f(i,j)$ on a grid; region of interest $\mathcal{R}$; measurement vector $y$; covariance matrix $\Sigma$; stride $s$; number of coarse seeds $K_0$.
\Statex \textbf{Output:} estimated cell $(\hat{i},\hat{j})$.

\State precompute $\Sigma^{-1}$
\State $\mathcal{C} \gets \emptyset$ \Comment{coarse candidates}

\For{each $(i,j) \in \mathcal{R}$ sampled on stride $s$}
    \State $D^2(i,j) \gets \bigl(y - f(i,j)\bigr)^\top \Sigma^{-1} \bigl(y - f(i,j)\bigr)$
    \State append $\bigl(D^2(i,j),\, i,\, j\bigr)$ to $\mathcal{C}$
\EndFor

\State sort $\mathcal{C}$ by $D^2$; let $\mathcal{S}$ be the first $K_0$ seeds
\State $\mathcal{R}_{\mathrm{ref}} \gets \emptyset$ \Comment{refined ROI}

\For{each $(\cdot,i_s,j_s)$ in $\mathcal{S}$}
    \For{ $(i,j)$ in a fixed local window around $(i_s, j_s)$}
        \State add $(i,j)$ to $\mathcal{R}_{\mathrm{ref}}$
    \EndFor
\EndFor

\State $\mathcal{C}_{\mathrm{ref}} \gets \emptyset$ \Comment{refined candidates}

\For{each $(i,j) \in \mathcal{R}_{\mathrm{ref}}$}
    \State $D^2(i,j) \gets \bigl(y - f(i,j)\bigr)^\top \Sigma^{-1} \bigl(y - f(i,j)\bigr)$
    \State append $\bigl(D^2(i,j),\, i,\, j\bigr)$ to $\mathcal{C}_{\mathrm{ref}}$
\EndFor

\State sort $\mathcal{C}_{\mathrm{ref}}$ by $D^2$; let $(D^2_{\min},\hat{i},\hat{j})$ be the first element
\State \textbf{return} $(\hat{i},\hat{j})$
\end{algorithmic}
\end{algorithm}

\begin{figure*}[ht!]
    \centering
    \begin{subfigure}[t]{0.49\textwidth}
        \centering
        \includegraphics[height=2.2in]{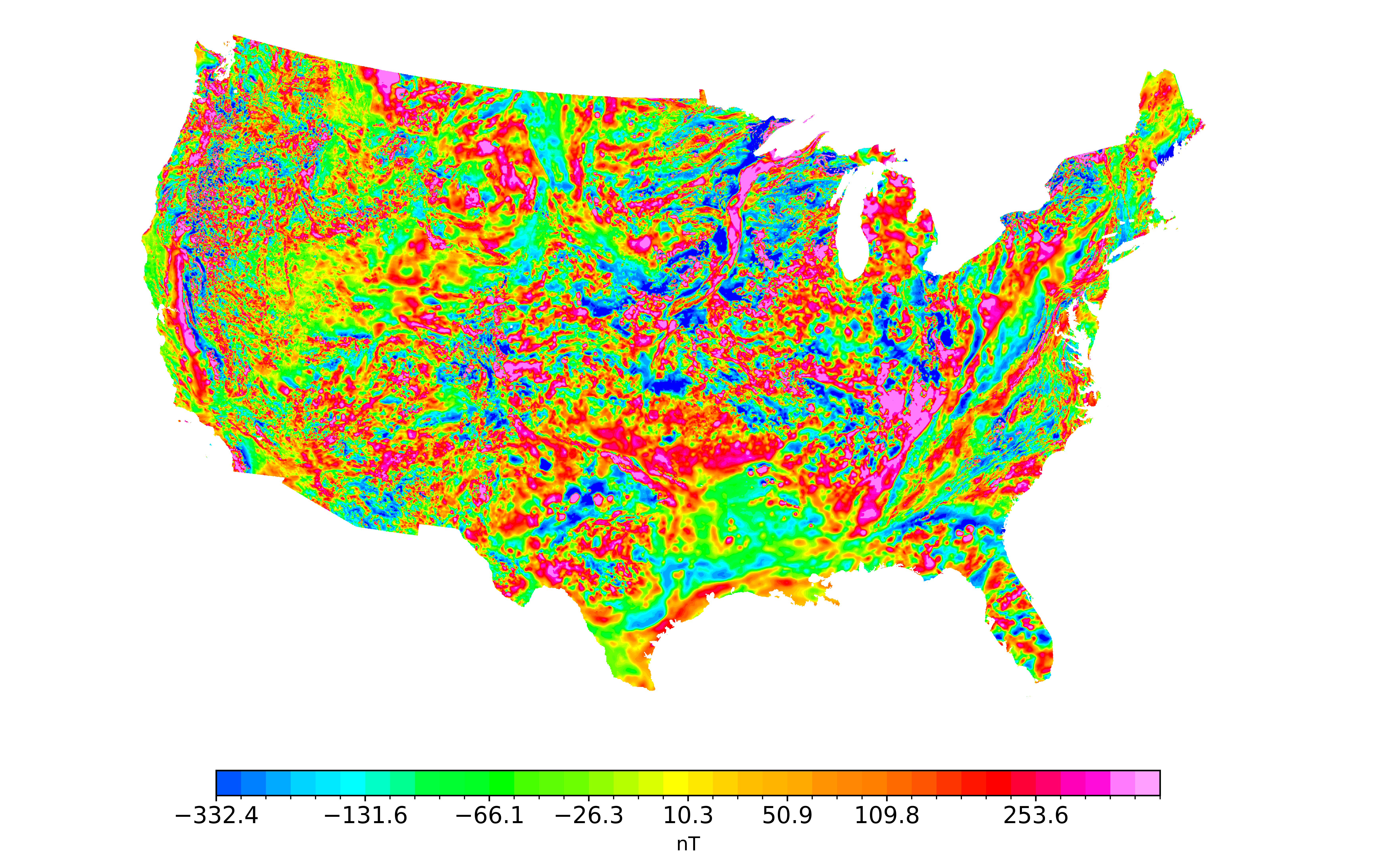}
        \caption{}
        \label{fig:roi_CONUS}
    \end{subfigure}
    \begin{subfigure}[t]{0.49\textwidth}
        \centering
        \includegraphics[height=2.2in]{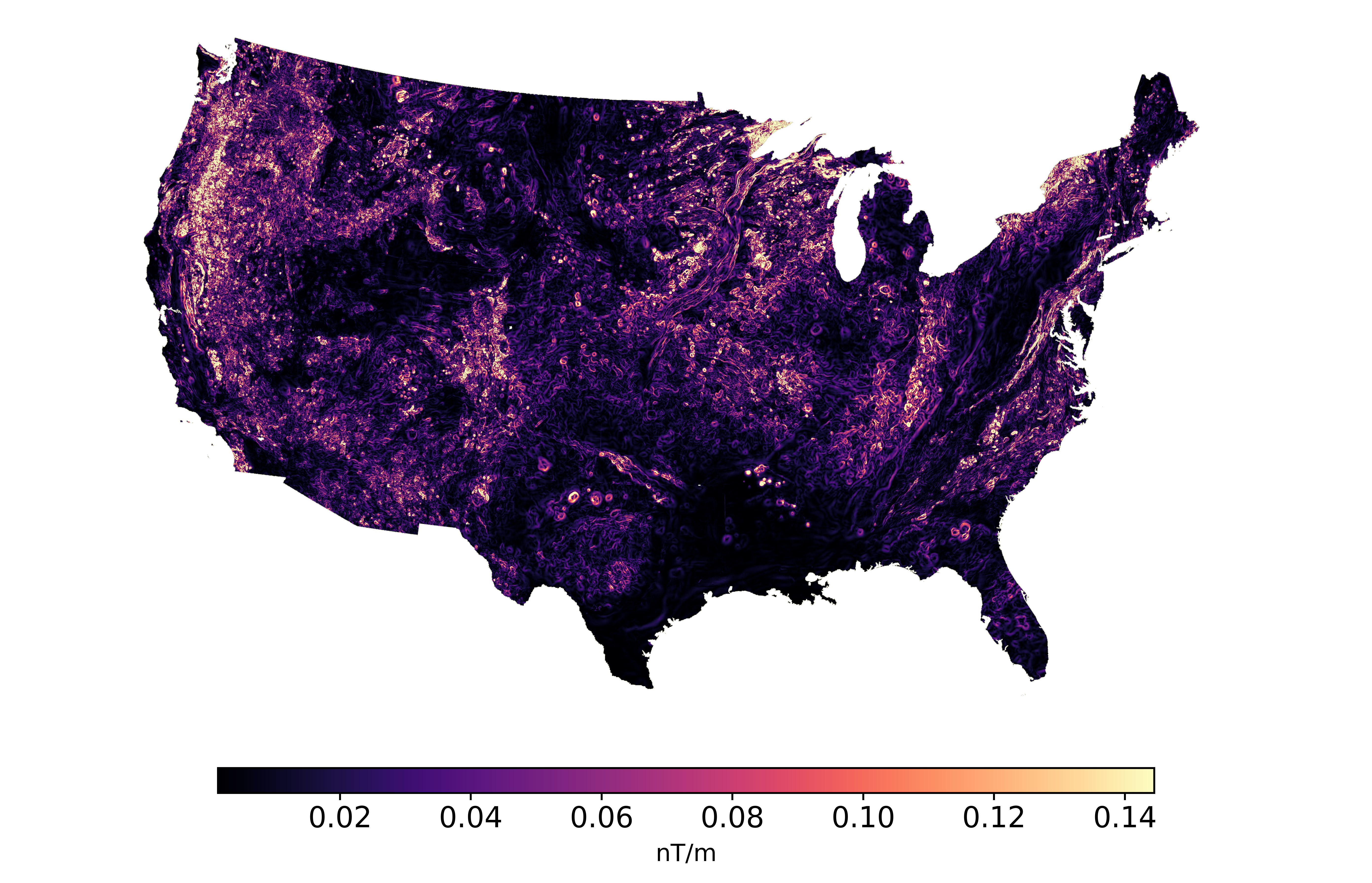}
        \caption{}
        \label{fig:roi_grad}
    \end{subfigure}
    
    \caption{Region of interest in our simulation. (a) Magnetic anomaly map of the Conterminous United States (CONUS) (b) Gradient magnitude of the magnetic anomaly field}
    \label{fig:ROI}
\end{figure*}

\noindent\textbf{Computational complexity.}
Let $N$ denote the total number of $2\times 2$ cells in the initial ROI, and let $d$ denote the dimension of the feature vector ($d=3$ for gradient space and $d=S\le 4$ for corner space). For fixed $d$, the cost of evaluating $D^2(i,j)$ at one cell is $O(1)$ once ${\Sigma^{-1}}$ has been precomputed.

\begin{itemize}
    \item \emph{Gradient precomputation.} The feature map $(g_x,g_y,d_{xy})$ is obtained by convolving the magnetic map with $2\times 2$ finite-difference kernels. This is a one-time $O(N)$ operation performed offline and amortized over all localization queries.
    \item \emph{Coarse pass.} The stride $s$ subsampling evaluates $D^2(i,j)$ on approximately $N/s^2$ cells. This cost is $O(N/s^2)$ for distance evaluations plus $O((N/s^2)\log(N/s^2))$ to sort and retain the $K_0$ smallest values.
    \item \emph{Refine pass.} The refined ROI $\mathcal{R}_{\mathrm{ref}}$ is constructed as the union of fixed local windows around the $K_0$ coarse seeds. If each window contains at most $C$ cells, then
    \begin{equation}
        |\mathcal{R}_{\mathrm{ref}}| \,\le\, C K_0,
    \end{equation}
    so the refine pass requires $O\bigl(|\mathcal{R}_{\mathrm{ref}}|\bigr) = O(1)$ distance evaluations in $N$.
\end{itemize}

Overall, for fixed stride $s$, number of seeds $K_0$, and fixed local window size, the time complexity per localization query scales as
\begin{equation}
    T_{\mathrm{search}}(N)
    = O\bigl(N/s^2\bigr) + O\bigl(|\mathcal{R}_{\mathrm{ref}}|\bigr)
    = O(N)
\end{equation}
dominated by the coarse pass. The space complexity is $O(N)$ for the precomputed features plus $O(K_0)$ for the seed list and $O(|\mathcal{R}_{\mathrm{ref}}|)$ for the refined candidates. 

\section{Simulation Results and Discussion\label{sec:simulation}}
In this section, we evaluate the end-to-end performance of the proposed quantum magnetic sensing based geo-localization system through numerical simulations over real-world geomagnetic data. Using the United States Geological Survey (USGS) composite magnetic anomaly grid for the United States and Canada, we construct a CONUS-region magnetic map and derive its gradient structure to characterize localization difficulty across different sites. We then quantify the accuracy of the two-stage adaptive sensing protocol when estimating the magnetic field strength contained in this USGS magnetic dataset. Next, we examine how localization performance varies with sensing budget and sensor count under both gradient-space and corner-space Mahalanobis search. Finally, we analyze the computational cost of the search as a function of the size of region of interest (ROI).

\subsection{Magnetic Map}
In this work, we used the composite magnetic anomaly grid for the United States and Canada published by the USGS in March 2025 \cite{mccafferty2023national}. This dataset merges magnetic surveys \cite{ravat2010preliminary, ADGGS2016, MilesOneschuk2016} to create a magnetic anomaly map of North America with a resolution of approximately 1 km (30 arc-seconds). To extract the ROI from the broader North America dataset, we applied a geospatial clipping operation. A polygon defining the CONUS boundary was constructed from the Natural Earth Admin-1 states and provinces dataset \cite{NaturalEarth2024}. This polygon was used to mask the original GeoTIFF, excluding Alaska, Hawaii, and Canadian territories. Fig.~\ref{fig:roi_CONUS} shows the resulting ROI, spanning longitudes from approximately \ang{-125} to  \ang{-66.5} and latitudes from \ang{24} to \ang{49.5}. The anomaly values within the CONUS ROI exhibit a wide dynamic range, from approximately $-332.4$ nT to over $430.4$ nT, reflecting the complex and varied geology of the continent.

To characterize the spatial rate of change of the magnetic anomaly field, we derived a gradient magnitude map from the ROI as shown in Fig. \ref{fig:roi_grad}. The gradient components ($g_x$ and $g_y$) were computed via numerical differentiation of the anomaly grid at cell centers, and the gradient magnitude was calculated as $||\nabla Z|| = \sqrt{g_x^2 + g_y^2}$. This map highlights local contrast in the magnetic anomaly field, with values spanning from 0.02 to 0.14 nT/m over ROI. For localization, these contrasts translate directly into the discriminability of neighboring cells. In high $||\nabla Z||$ terrain, adjacent map cells have more distinct signatures, yielding larger Mahalanobis separation and thus more accurate position estimates. Conversely, in low $||\nabla Z||$ terrain, many cells are nearly indistinguishable, which increases positioning error.

When simulating physical magnetometer measurements, we constructed total field magnetic map, which is the sum of the core Earth's field (main field) and the magnetic anomaly field. We synthesized this total field map by augmenting the USGS anomaly grid with the International Geomagnetic Reference Field (IGRF) model. The geodetic coordinates (latitude, longitude) of each pixel center were transformed into the magnetic field vector, from which the total intensity $B_{IGRF}$ was extracted. The final total field value $B_{total}$ for each pixel was then calculated by summing the crustal anomaly $B_{anomaly}$ and the main field:
\begin{equation}
    B_{total} = B_{IGRF} + B_{anomaly}
\end{equation}
The resulting total field map retains the high-frequency content of the geological anomalies but is dominated by the smooth wavelength trend of the Earth's main field.

\subsection{Magnetic Field Estimation}
\label{subsec:field_estimate}
\begin{figure}[t]
\centering
    \includegraphics[width=\columnwidth]{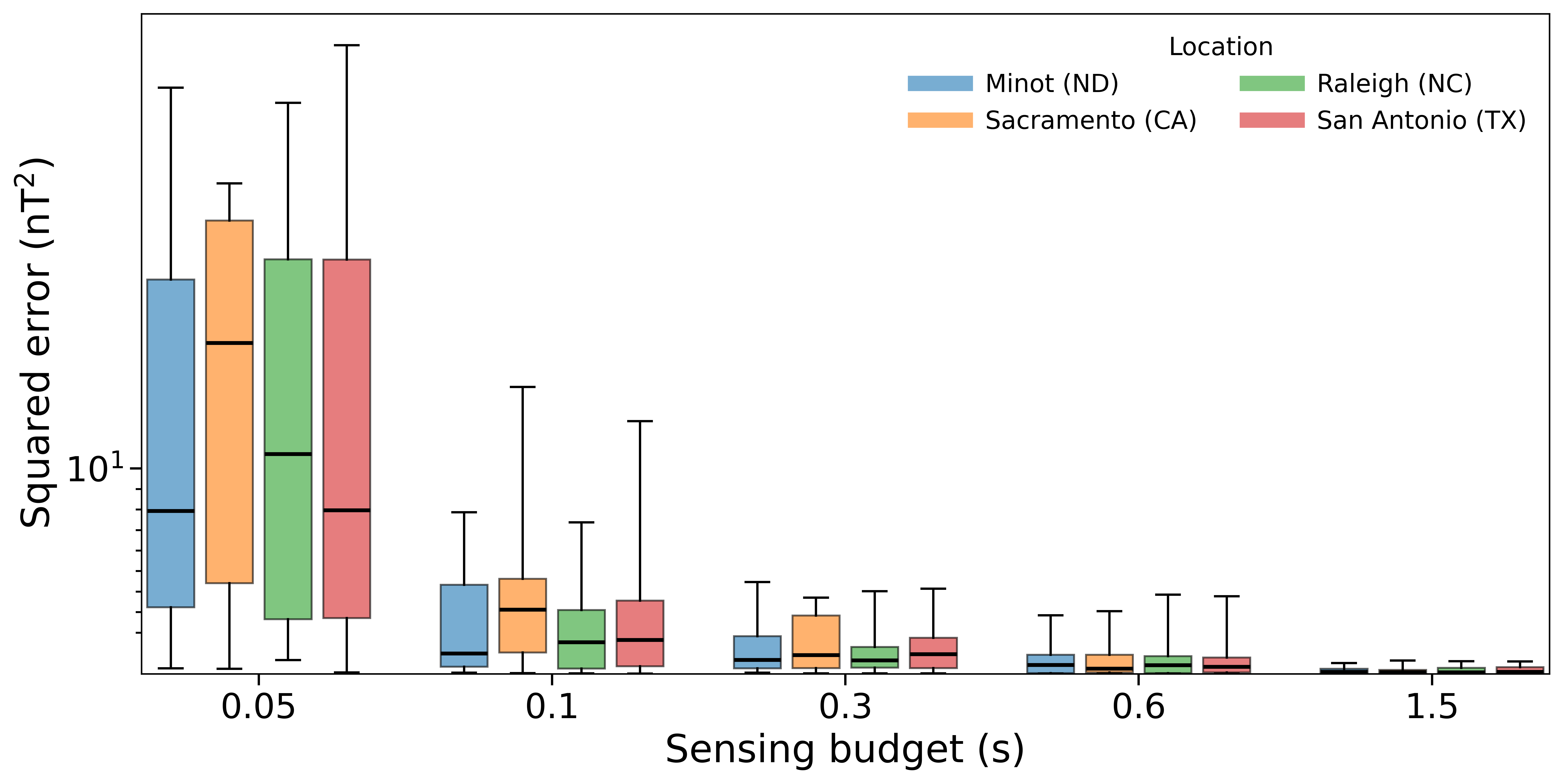}
    \caption{ Square error (SE) of magnetic field estimates as a function of sensing budget for four CONUS sites}
    \label{fig:mse_vs_sensing_budget}
\end{figure}

The magnetic field estimation accuracy was evaluated across four geographically diverse CONUS locations (Minot-ND, Raleigh-NC, Sacramento-CA, San Antonio-TX). At each site, we simulated NV-center magnetometry measurements at the four corners of a $2 \times2$ grid cell  with ground truth magnetic field values from total field maps. The protocol was executed across five sensing budgets (0.05, 0.10, 0.30, 0.60, and 1.5 seconds) with 10 independent repetitions per site-budget pair, resulting in 800 total observations across the experimental matrix. Estimation accuracy was quantified using the squared error (SE) between the estimated magnetic field ($\hat{B}$) and the ground truth ($B_{true}$):
\begin{equation}
    SE_B = (\hat{B} - B_{true})^2
\end{equation}
where all magnetic field values are in nT.

Fig.~\ref{fig:mse_vs_sensing_budget} presents the distribution of SE values with increasing budget across all sites. At the minimal budget of 0.05 seconds, median SE values range from approximately 5-95 $\text{nT}^2$ across locations. This reflects the inherent uncertainty in rapid measurements where the protocol has insufficient time for precise frequency estimation. As the budget increases to 1.50 seconds, median SE values decrease dramatically to approximately 0.01-0.5 $\text{nT}^2$, representing a 2-3 order of magnitude improvement.  This demonstrates the protocol's ability to achieve high precision magnetic field estimates given sufficient sensing time, with interquartile ranges tightening considerably. Minot-ND shows the most rapid convergence to low error values, achieving sub-1 $\text{nT}^2$ median SE by 0.30 seconds, whereas Sacramento-CA exhibits the most challenging conditions, with several extreme outliers at low budgets. 

The performance variations between locations stem from differences in local magnetic field gradients and total field intensities (ranging from $~45,600 - 55,300$ nT across locations). The consistent separation between budget groups across all locations, combined with the tight clustering of results at higher budgets, demonstrates statistically significant improvements in estimation accuracy with increased sensing time.

\begin{figure*}[ht!]
    \centering
    \begin{subfigure}[t]{0.49\textwidth}
        \centering
        \includegraphics[width = \linewidth, height=2.005in]{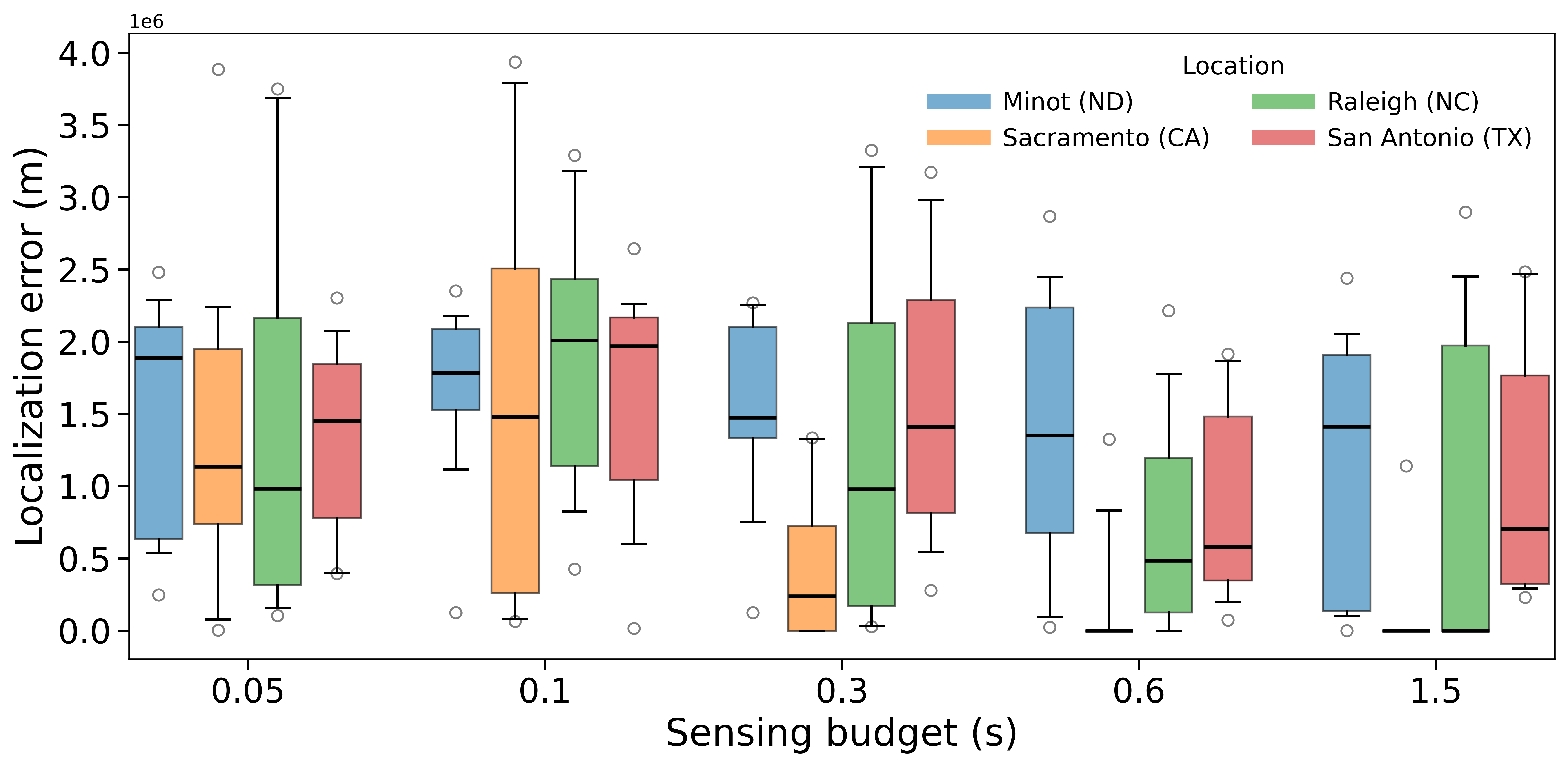}
        \caption{}
        \label{fig:vs_sensing_budget}
    \end{subfigure}
    \hfill
    \begin{subfigure}[t]{0.49\textwidth}
        \centering
        \includegraphics[width = \linewidth, height=2in]{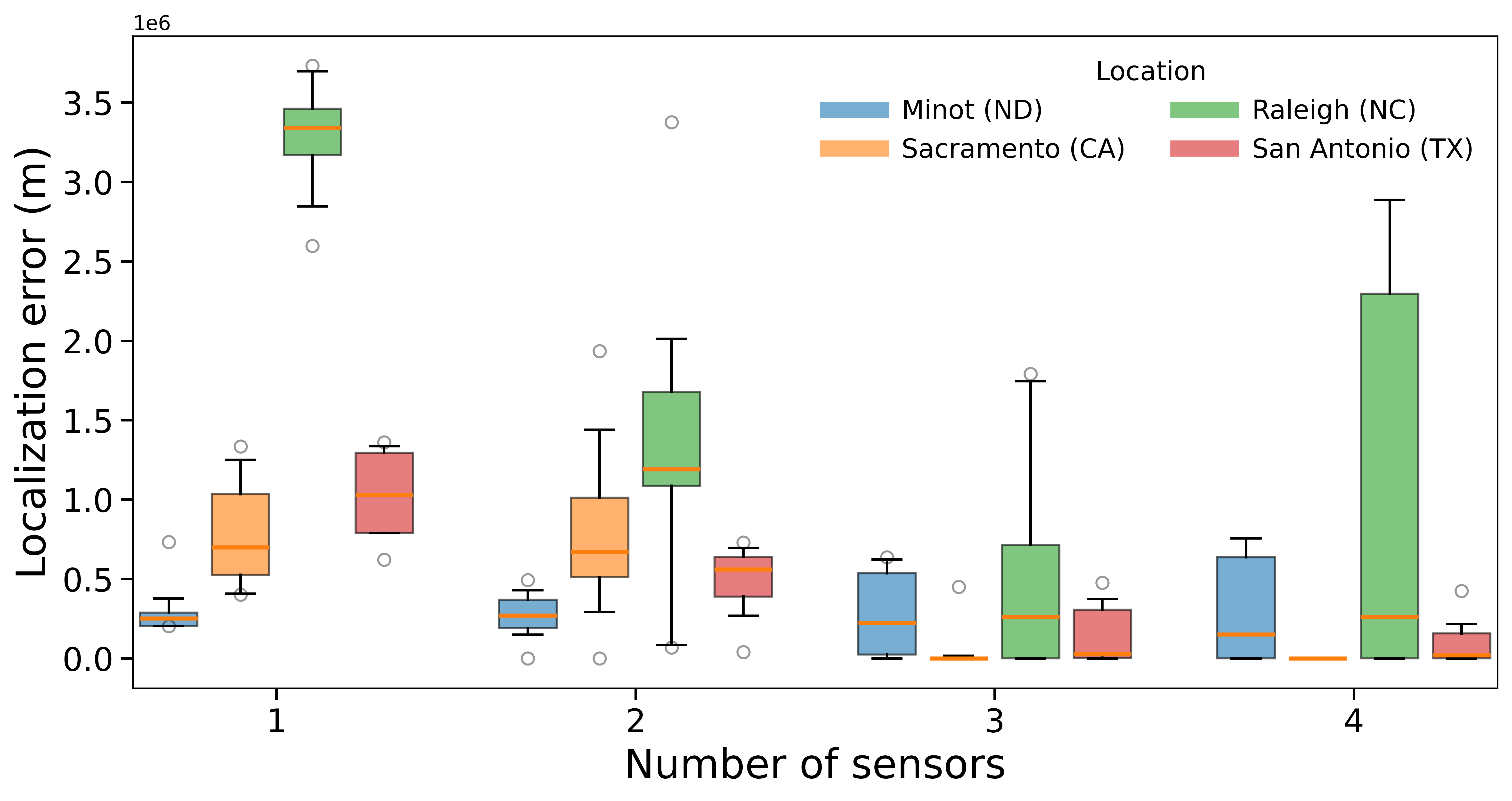}
        \caption{}
        \label{fig:vs_nsensors}
    \end{subfigure}
    \hfill
    \begin{subfigure}[t]{0.49\textwidth}
        \centering
        \includegraphics[width = \linewidth, height=2in]{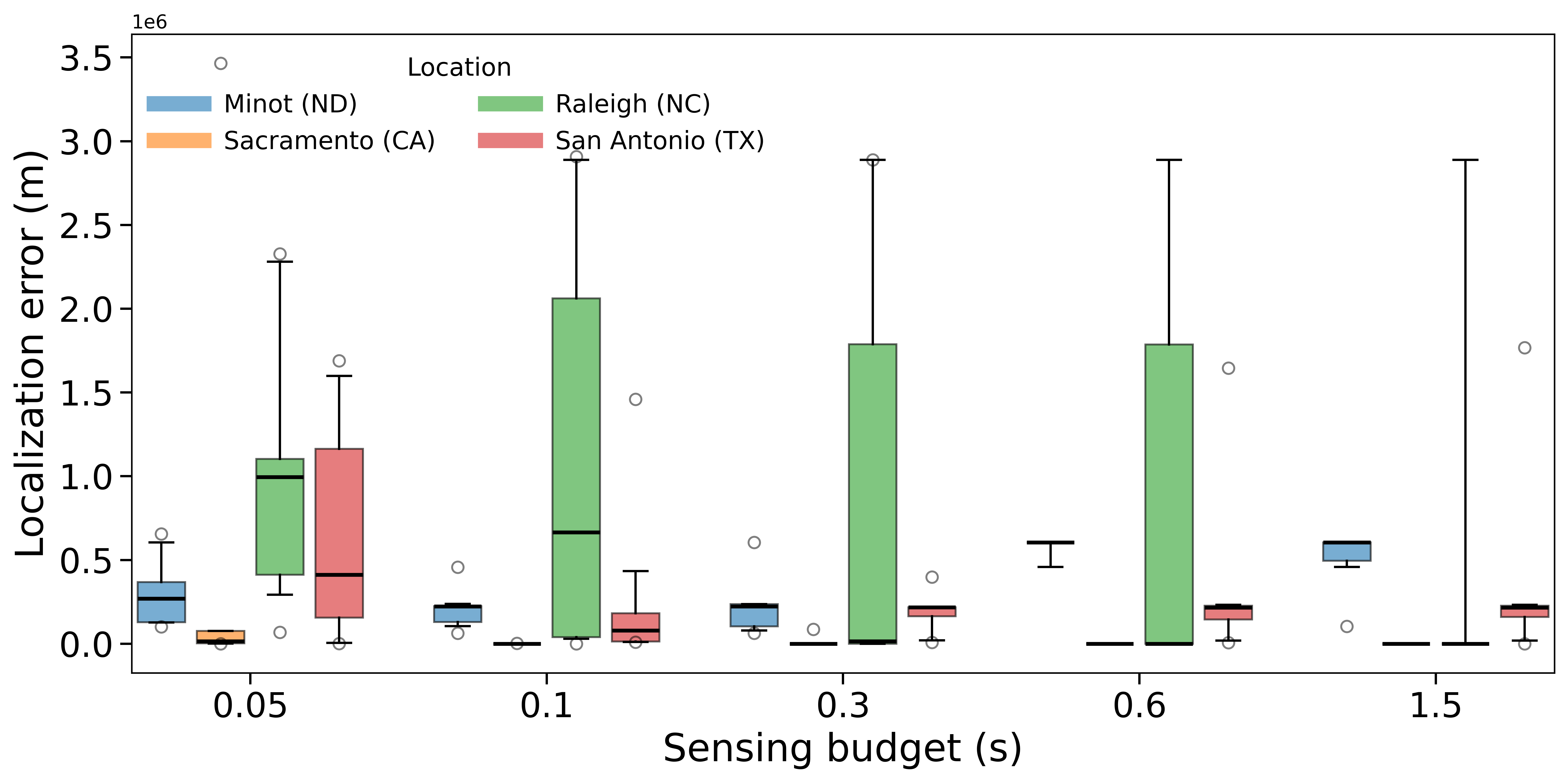}
        \caption{}
        \label{fig:grad_corner}
    \end{subfigure}
    \hfill
    \begin{subfigure}[t]{0.49\textwidth}
        \centering
        \includegraphics[width = \linewidth, height=2in]{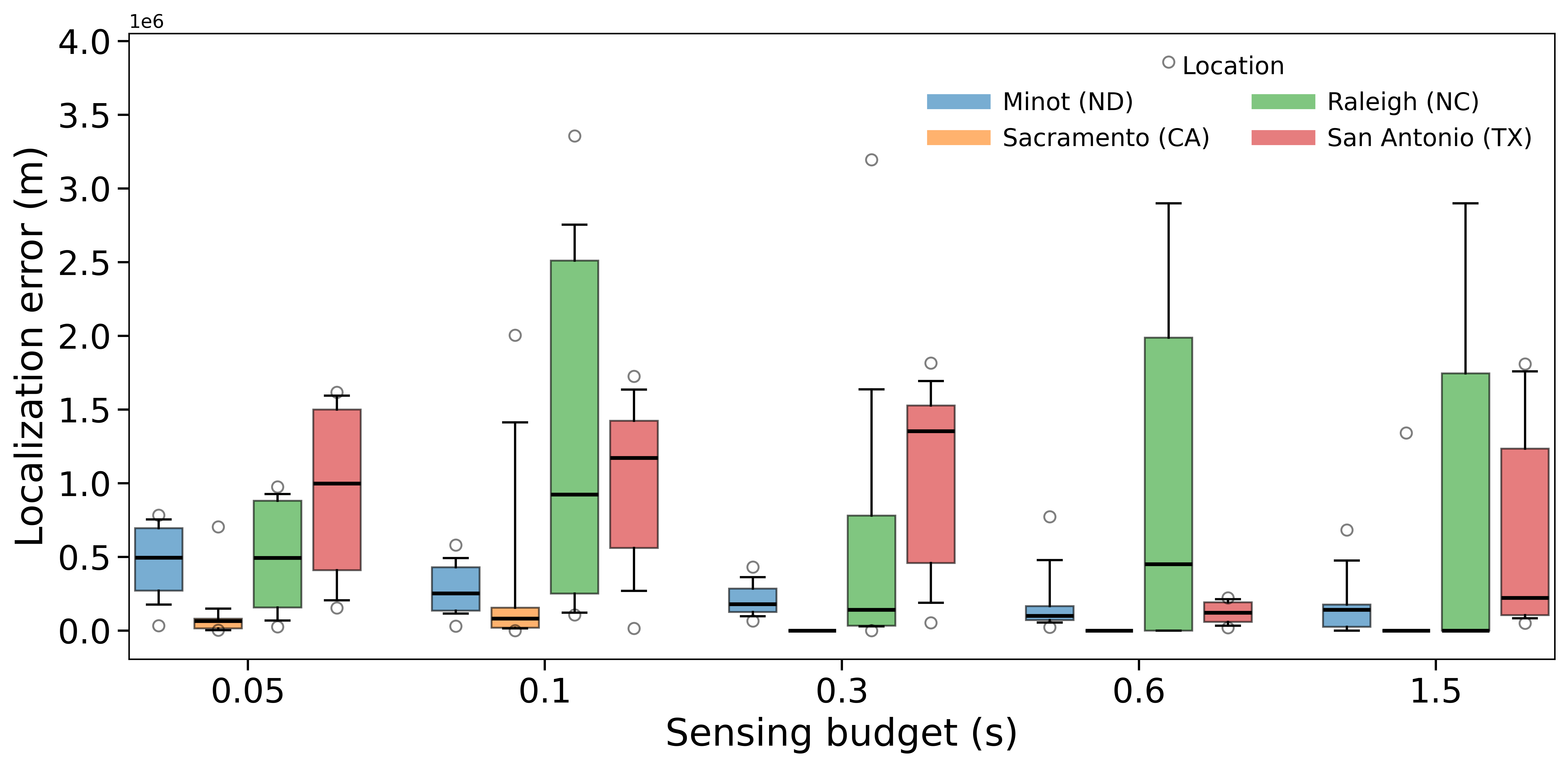}
        \caption{}
        \label{fig:corner_grad}
    \end{subfigure}
    \caption{Localization performance across four CONUS sites. (a)  Localization error as a function of sensing budget for gradient-space Mahalanobis localization. (b) Localization error as a function of the number of sensors for a fixed total sensing budget of 1.5~s using corner-space Mahalanobis localization. (c)--(d) Localization error as a function of sensing budget for gradient-space (corner-space) coarse pass followed by corner-space (gradient-space) refinement.
 }
\end{figure*}


\subsection{Localization error}
\label{subsec:loc_error}

\begin{figure}[ht!]
\centering
    \includegraphics[height = 2.4 in, keepaspectratio]{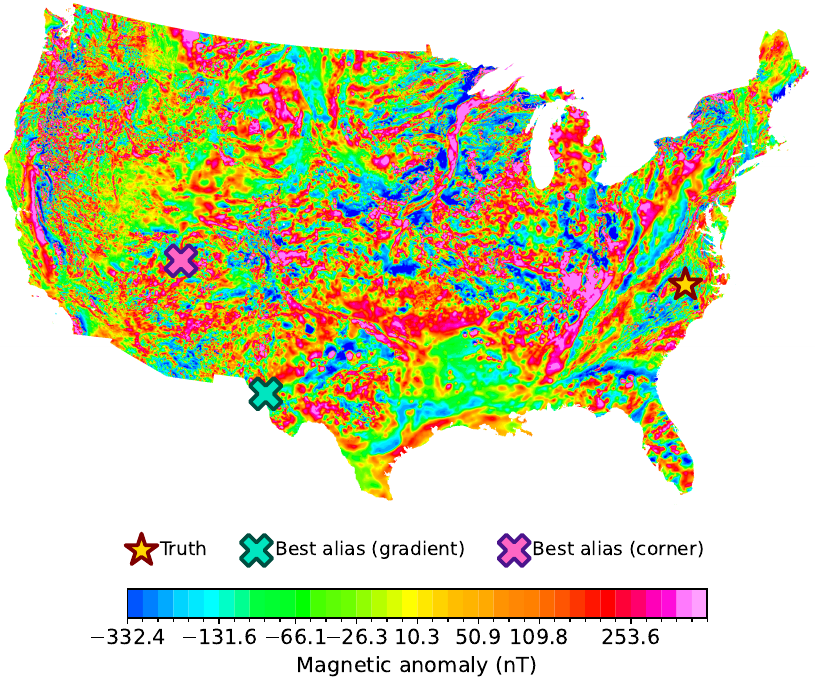}
    \caption{The true test location in Raleigh-NC and the globally best-matching alias locations in gradient space and corner space.}
    \label{fig:ambiguity}
\end{figure}
We evaluated localization performance across four CONUS sites under three Mahalanobis-based search configurations: (i) gradient-space search with varying sensing budgets, (ii) corner-space search with varying numbers of sensors at a fixed total sensing budget, and (iii) a two-metric search that combines gradient and corner features. In all cases, the localization error is defined as the Euclidean distance between the true and estimated positions in the map coordinate system. Fig.~\ref{fig:vs_sensing_budget} reports localization performance across five total sensing budgets ($0.05-1.50$~s) using four sensors and gradient-space Mahalanobis search, where gradient features are used for both the coarse and fine searches. At the smallest budget of $0.05$~s, median localization errors are on the order of $10^6$~m across all sites, reflecting the limited information available from very limited number of sensing events. As the budget increases to $1.5$~s, Raleigh and Sacramento, which lie in high-gradient regions, achieve essentially zero median error. In contrast, Minot and San Antonio, located in smoother magnetic regions, still exhibit residual kilometer-scale errors despite the longer sensing time. This behavior highlights the joint role of sensing time and local gradient structure in determining achievable localization accuracy. Fig.~\ref{fig:vs_nsensors} examines the effect of sensing geometry using corner-space Mahalanobis search at a fixed sensing budget of $1.5$~s. With a single sensor, median errors range from hundreds of kilometers (Minot) to several million meters (Raleigh), with broad error distributions. As additional sensors are added, the error distributions contract, and the medians drop systematically. With four sensors, Sacramento and San Antonio achieve approximately zero median error, while Minot and Raleigh show order-of-magnitude improvements relative to the single sensor case. This trend confirms that spatial diversity in the corner measurements provides a strong geometric lever for disambiguating candidate map cells, particularly when the sensing budget is fixed.

Fig.~\ref{fig:grad_corner} and Fig.~\ref{fig:corner_grad} quantify the performance of the two-metric Mahalanobis pipelines, which combine gradient and corner space but differ in the order in which these features are used. In Fig.~\ref{fig:grad_corner}, the algorithm first performs a coarse pass in gradient space and then refines in corner space. Sacramento and Raleigh again show the best performance, with a zero-median for budgets $\ge 0.3$~s. Minot and San Antonio benefit from the two-metric scheme but remain limited by the weak local gradients. Fig.~\ref{fig:corner_grad} inverts the order, using corner-space coarse pass followed by gradient-space refinement. This variant yields slightly improved medians for Minot and San Antonio, especially at intermediate budgets, while maintaining low error for Sacramento and Raleigh. Across both figures, the two-metric pipelines reduce large outliers relative to the single-metric gradient search and provide more stable performance across sites, indicating that corner and gradient features capture complementary aspects of the magnetic map.

Throughout all four localization experiments in Fig.~\ref{fig:vs_sensing_budget}–\ref{fig:corner_grad}, sites in high-gradient regions (Sacramento-CA and Raleigh-NC) achieve near-zero median error once the sensing budget or sensor count is sufficient. Smoother regions (Minot-ND and San Antonio-TX) remain constrained by weaker magnetic spatial structure, even with longer sensing times or additional sensors. Gradient-space Mahalanobis search benefits most from increased sensing time, while corner-space search gains substantially from multiple sensors. The two-metric pipelines further suppress large outliers and stabilize performance across sites by combining these complementary metrics. We further look at the large interquartile range of Raleigh's localization performance using Fig.~\ref{fig:ambiguity}, which depicts the true test location and the best-matching aliases in gradient and corner feature spaces. Although the gradient and corner feature vectors at these alias cells are nearly indistinguishable from those at the true cell, the corresponding locations lie roughly $2{,}492$~km and $2{,}909$~km away from Raleigh, respectively. Because such distant locations look almost identical in feature space, sensor noise in the NV-center magnetometer can occasionally steer the Mahalanobis search toward one of these aliases, inflating the interquartile range of the localization error.

\subsection{Mahalanobis Runtime}
\begin{figure}[ht!]
\centering
    \includegraphics[width=\columnwidth, height = 2.43in, keepaspectratio]{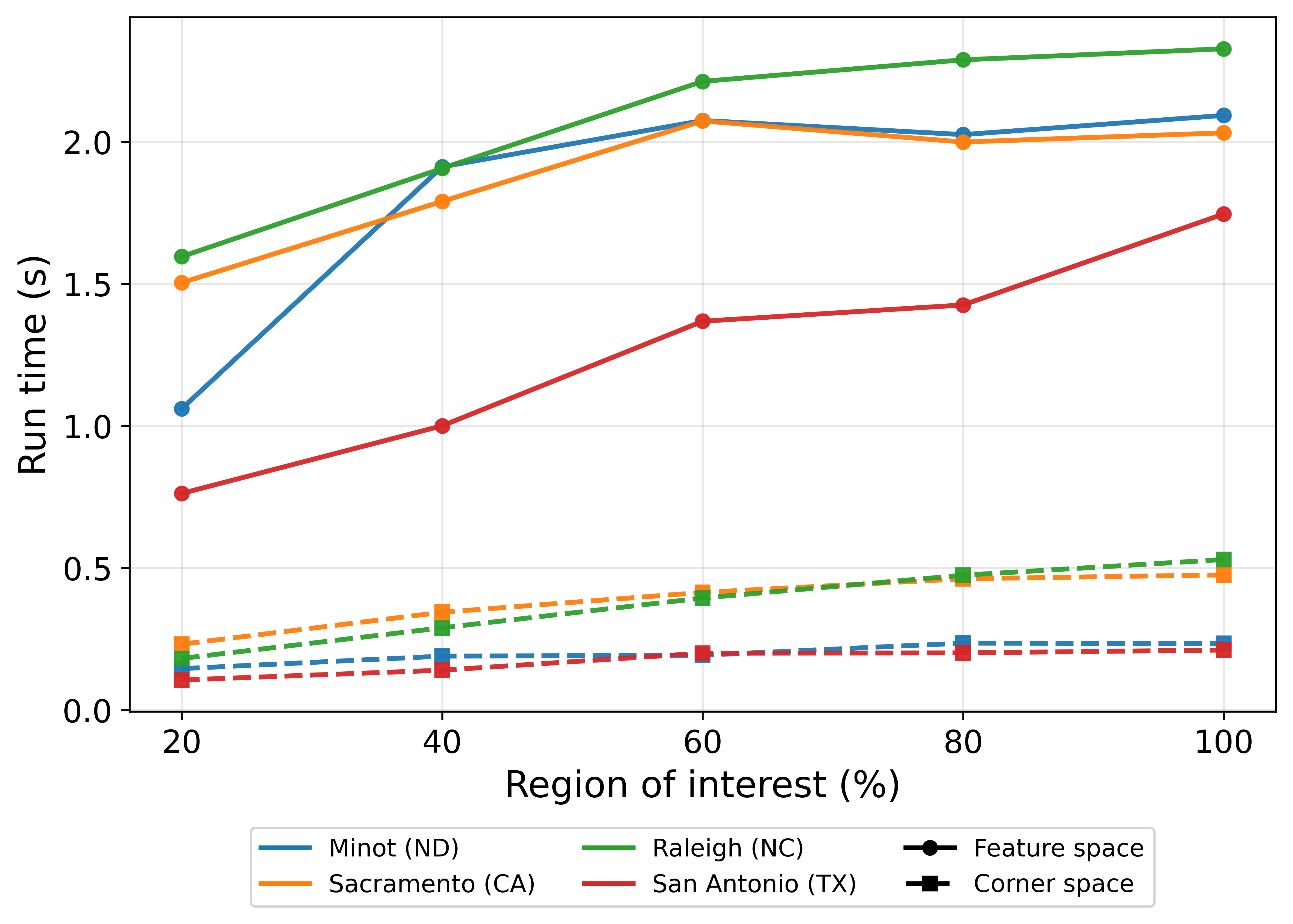}
    \caption{Runtime of the Mahalanobis search as a function of ROI size for four CONUS locations.}
    \label{fig:mahal_runtime}
\end{figure}

To evaluate the computational cost of the Mahalanobis search, we measured runtime as a function of ROI size for all four CONUS locations as shown in Fig.~\ref{fig:mahal_runtime}. All simulations used a fixed sensing budget of $1.5$~s and four sensors for both the feature-space and corner-space formulations. For each location, we constructed a square ROI centered at the true cell, whose size ranged from $20$–$100\%$ of the original CONUS map. Across all locations, feature-space search exhibits a clear increase in runtime with ROI size.  For the smallest ROI, median runtimes lie between about $0.75$~s (San Antonio) and $1.6$~s (Raleigh).  When the ROI is expanded to the full map, the runtimes increase to approximately $2.0$–$2.4$~s. This sublinear growth with ROI size is consistent with the design of the search algorithm.

Corner-space Mahalanobis search shows the same qualitative dependence on ROI size, but with a substantially smaller factor.  Across all locations and all ROI sizes, corner-space runtimes remain between $0.1$~s and $0.6$~s, achieving a speedup of roughly $4$–$8 \times$ relative to feature-space search. These results show that the proposed Mahalanobis search remains computationally tractable even when scanning the entire CONUS map. The corner representation is therefore attractive not only for its favorable localization performance in smooth magnetic regions, but also for its significantly reduced search cost, which becomes important when running many localization queries or deploying the method in real-time systems.

\section{Conclusion\label{sec:conclusion}}

This paper investigates the feasibility of quantum-enhanced, map-based geo-localization using quantum magnetic sensing. We first developed a unified problem formulation that separates the geometric, sensing, and inference components of magnetic localization. Within this framework, we derived the Cramér–Rao lower bound for a Ramsey-based NV-center magnetometer with realistic readout fidelities and coherence times, and compared it to the fundamental precision limit of a classical scalar magnetometer. Building on this analysis, we integrated a two-stage adaptive sensing protocol with a Mahalanobis map-matching engine operating on a high-resolution USGS magnetic anomaly map of CONUS. 

Simulation results across four geographically diverse CONUS sites validate our integrated pipeline, demonstrating nT-level magnetic field estimation accuracy and sub-kilometer median localization error under realistic sensing budgets. In magnetically smoother areas, the corner-space formulation provides better localization accuracy compared to gradient-space search while maintaining computational tractability, even when scanning the full CONUS map. These findings collectively indicate that NV-center magnetometers, coupled with appropriate adaptive sensing and map-matching algorithms, offer a viable and efficient path toward geo-localization in infrastructure-free environments.

\section*{Acknowledgment} \label{sec:ack}
This work is supported in part by the National Science Foundation under grants 2304118 and 2326746.

\appendices
\section{Proof of Theorem 1: Classical Magnetometer Cramér–Rao Bound}\label{appen:A}
We derive the Cramér–Rao lower bound for a classical scalar magnetometer as stated in Theorem. \ref{thm:1}. The measurement model is give by:
\begin{equation}
    \hat{B} = B + \epsilon,
\end{equation}
where $\epsilon$ is drawn from a Gaussian distribution with zero mean and variance $\sigma^2$. The likelihood function of observing $\hat{B}$ given the true value $B$ is therefore:
\begin{equation}
    p(\hat{B}|B) = \frac{1}{\sqrt{2\pi\sigma^2}} \exp\left(-\frac{(\hat{B} - B)^2}{2\sigma^2}\right).
\end{equation}
The Fisher information $F(B)$ for a single parameter is defined as \cite{dammann2024cramer}:
\begin{align}
    F(B) &= \mathbb{E}\left[\left(\frac{\partial}{\partial B} \ln p(\hat{B}|B) \right)^2 \right] = \mathbb{E}\left[\left( \frac{\hat{B} - B}{\sigma^2} \right)^2 \right] \\ \nonumber
    &= \frac{1}{\sigma^4} \mathbb{E}\left[ (\hat{B} - B)^2 \right] = \frac{1}{\sigma^4} \cdot \sigma^2 = \frac{1}{\sigma^2}.
\end{align}

According to the Cramér–Rao inequality, the variance of any unbiased estimator $\hat{B}$ of $B$ satisfies:
\begin{equation}
    Var_C(B) \ge \frac{1}{F(B)} = \sigma^2.
\end{equation}
which completes the proof.
\balance
\bibliographystyle{IEEEtran}
\bibliography{references}

\end{document}